\documentclass[11pt]{article}

\usepackage{a4}
\usepackage{cite}
\usepackage{paralist}
\usepackage{graphicx}
\usepackage{subfigure}
\usepackage{amssymb}
\usepackage{amsmath}
\allowdisplaybreaks[1]

\usepackage{amsthm}

\newtheorem{theorem}{Theorem}
\newtheorem{lemma}[theorem]{Lemma}

\newtheorem{claim}[theorem]{Claim}

\makeatletter
\def\@endtheorem{\endtrivlist}
\makeatother

\newcommand{\Deg}[1]{\deg(#1)}
\newcommand{\Degb}[2]{\deg_{#1}(#2)}
\newcommand{\Rule}[3]{(#1,#2)}

\begin{document}

\title{Parameterized algorithm for 3-path vertex cover}
\author{Dekel Tsur%
\thanks{Department of Computer Science, Ben-Gurion University of the Negev.
Email: \texttt{dekelts@cs.bgu.ac.il}}}
\date{}
\maketitle

\begin{abstract}
In the $3$-path vertex cover problem,
the input is an undirected graph $G$ and an integer $k$.
The goal is to decide whether there is a set of vertices $S$ of
size at most $k$ such that every path with $3$ vertices in $G$ contains at
least one vertex of $S$.
In this paper we give parameterized algorithm for $3$-path cover
whose time complexity is $O^*(1.713^k)$.
Our algorithm is faster than previous algorithms for this problem.
\end{abstract}


\section{Introduction}
For an undirected graph $G$, an \emph{$l$-path vertex cover} is a set of
vertices $S$ such that every path with $d$ vertices in $G$ contains at
least one vertex of $S$.
In the \emph{$l$-path vertex cover problem},
the input is an undirected graph $G$ and an integer $k$. The goal is to decide
whether there is an $l$-path vertex cover of $G$ with size at most $k$.
The problem for $l=2$ is the famous vertex cover problem.
The problem is NP-hard for every constant $l \geq 2$.

In this paper we consider the 3-path vertex cover problem.
The first non-trivial parameterized algorithm for the 3-path vertex cover
problem was given by Tu~\cite{tu2015fixed},
which gave an $O^*(2^k)$-time algorithm.
This result was subsequently improved several times:
Wu~\cite{wu2015measure} gave an $O^*(1.882^k)$-time algorithm,
Katreni\v{c}~\cite{katrenivc2016faster} gave an $O^*(1.818^k)$-time algorithm,
and Chang et al.~\cite{chang2016fixed} gave an $O^*(1.749^k)$-time algorithm,
and Xiao and Kou~\cite{xiao2017kernelization} also gave an $O^*(1.749^k)$-time
algorithm.
The algorithm of Chang et al.\ uses exponential space while the algorithm
of Xiao and Kou uses polynomial space.
The non-parameterized 3-path vertex cover problem has also been
studied~\cite{kardovs2011computing,chang20141,xiao2017exact}.
The fastest exact algorithm for this problem runs in $O^*(1.366^n)$
time~\cite{xiao2017exact}.

In this paper, we give an algorithm whose time complexity is $O^*(1.713^k)$.
This improves the previous parameterized algorithms for this problem.
Our algorithm is based on the algorithm of Xiao and
Kou~\cite{xiao2017kernelization}. In Section~\ref{sec:Xiao} we describe
the algorithm of Xiao and Kou, and in Section~\ref{sec:new} we describe
our algorithm.

\section{Preliminaries}
For a graph $G=(V,E)$ and a vertex $v\in V$,
$N_G(v) = \{u \in V\colon (u,v)\in E\}$,
$N_G[v] = N_G(v) \cup \{v\}$, and
$\Degb{G}{v} = |N_G(v)|$.
For a set of vertices $S$,
$N_G(S) = (\bigcup_{v\in S} N_G[S]) \setminus S$ and
$N_G[S] = \bigcup_{v\in S} N_G[v]$.

For a graph $G=(V,E)$ and a set of vertices $S$, $G[S]$ is the subgraph
of $G$ induced by $S$ (namely, $G[S]=(S,E\cap (S\times S)$).
We also define $G-S = G[V\setminus S]$.
For a set that consists of a single vertex $v$, we write $G-v$ instead of
$G-\{v\}$.

A vertex $v$ \emph{dominates} a vertex $u$ if $(u,v)\in E$ and
$N(u) \subseteq N[v]$. We also say that $v$ is a \emph{dominating vertex}.
A vertex $v$ \emph{weakly dominates} a vertex $u$ if $(u,v)\in E$ and
there is a vertex $x \in N(u)$ such that $x\notin N[v]$ and
$N(u)\setminus\{x\} \subseteq N[v]$.

\begin{claim}\label{clm:dominating}
If $u,v$ are adjacent vertices and $v$ does not dominate $u$ then
$|N(\{u,v\})| \geq \Deg{v}$. 
\end{claim}

A \emph{chain} is a path $x,x_1,x_2,x_3$ such that $x\neq x_3$,
$\Deg{x} \geq 3$, and $\Deg{x_1} = \Deg{x_2} = 2$.

\begin{claim}\label{clm:chain}
If $u,v$ are adjacent vertices without a common neighbor and
$\Deg{u} = \Deg{v} = 2$ then
$G$ contains a chain unless $G$ is a cordless path or a cordless cycle.
\end{claim}

\section{The algorithm of Xiao and Kou}\label{sec:Xiao}

In this section we describe a slightly modified version of 
the algorithm of Xiao and Kou.
A \emph{branching algorithm} is a recursive algorithm for a parameterized
problem that uses \emph{reduction rules} $R_1,\ldots,R_r$.
Each reduction rule $R_i$ has the form
$(G,k) \to (G_{i,1},k-c_{i,1}), \ldots, (G_{i,s_i},k-c_{i,s_i})$.
Additionally, every rule (except the last) has a condition which needs to be
satisfied in order for the rule to be applicable.
Given an instance $(G,k)$ for the problem, the algorithm chooses a rule $R_i$
and then recurse on the instances
$(G_{i,1},k-c_{i,1}), \ldots, (G_{i,s_i},k-c_{i,s_i})$.
The algorithm ``yes'' if and only if at least one recursive call returned
``yes''.
The recursion stops on an instance $(G',k')$ if either $G'$ is an empty graph
or $k' \leq 0$, and the algorithm returns either ``yes'' or ``no'' depending
whether $(G',k')$ is a yes instance.
To analyze the time complexity of the algorithm, define $T(k)$ to be the
maximum number of leaves in the recursion tree of the algorithm when the
algorithm is run on an instance with parameter $k$.
Each reduction rule $R_i$ with $s_i \geq 2$ corresponds to the following
recurrence on $T(k)$:
\[ T(k) \leq T(k-c_{i,1})+T(k-c_{i,2})+\cdots+T(k-c_{i,s_i}). \]
The largest root of the polynomial $P(x) = 1-\sum_{j=1}^{s_i} x^{-c_{i,j}}$
is called the \emph{branching factor} of the recurrence.
Let $\gamma$ be the maximum branching factor over all rules.
Then, the time complexity of the algorithm is $O^*(\gamma^k)$.

The algorithm of Xiao and Kou uses the following basic reduction rules.
\begin{enumerate}
\item[(B1)]
Let $v$ be a vertex.
$(G,k)\to
\Rule{G-v}{k-1}{\{v\}}$,
$\Rule{G-N[v]}{k-\Deg{v}}{N(v)}$, 
$\{\Rule{G-N[\{u,v\}]}{k-|N(\{u,v\})|}{N(\{u,v\})}\}_{u\in N(v)}$.
\item[(B2)]
Let $v$ be a dominating vertex.
$(G,k)\to
\Rule{G-v}{k-1}{\{v\}}$,
$\Rule{G-N[v]}{k-(\Deg{v}-1)}{N(v)\setminus\{u\}}$.
\item[(B3)]
Let $v$ be a weakly dominating vertex.
$(G,k)\to
\Rule{G-v}{k-1}{\{v\}}$, 
$\{\Rule{G-N[\{u,v\}]}{k-|N(\{u,v\})|}{N(\{u,v\})}\}_{u\in N(v)}$.
\item[(B4)] Let $v$ be a vertex with degree~3, whose neighbors $u_1,u_2,u_3$
satisfy $\Deg{u_1} = 1$ and $(u_2,u_3)\in E$.
$(G,k)\to
\Rule{G-N[v]}{k-2}{\{u_2,u_3\}}$, 
$\Rule{G-(N[\{u_2,u_3\}]\cup\{u_1\}}{k-|N(\{u_2,u_3\})|}{N(\{u_2,u_3\})}$.
\end{enumerate}

We now describe the reduction rules of the algorithm, which are also called
steps.
Given an instance $(G,k)$, the algorithm apply the first applicable rule.
Some rules are described as the application of two rules. This can be viewed
as an application of a single reduction rule.
For example, if we have two rules $(G,k) \to (G-S_1,k-c_1),(G-S_2,k-c_2)$
and $(G,k) \to (G-S_3,k-c_3),(G-S_4,k-c_4)$, then applying the first rule
on $(G,k)$ and then applying the second rule on the instance $(G-S_1,k-c_1)$
is equivalent to the rule
$(G,k) \to (G-(S_1 \cup S_3),k-c_1-c_3),
(G-(S_1 \cup S_4),k-c_1-c_4),(G-S_2,k-c_2)$.

\begin{enumerate}
\item[(S1)]
If $S$ is a connected component of $G$ such that $G[S]$ has maximum degree at
most $2$, apply the rule $(G,k)\to \Rule{G-S}{k-\gamma}{C}$ where
$\gamma$ is the minimum size of a 3-path vertex cover of $G[S]$.
\item[(S2)]
If $v$ is a vertex such that $\Deg{v} = 1$ and the neighbor $u$ of $v$ has
degree~2, apply the rule
$(G,k) \to \Rule{G-N[u]}{k-1}{N(u)\setminus\{v\}}$.
\item[(S3)]
If $v$ is a dominating vertex and $\Deg{v}\geq 3$,
apply (B2) on $(G,k)$ and $v$.
\item[(S4)]
If $v,v_1,v_2,v_3$ is a chain in $G$, apply (B3) on $(G,k)$ and $v$.
Apply (S2) on $(G-v,k-1)$ and $v_1$.
\item[(S5)]
If $v$ is a weakly dominating vertex and $\Deg{v} \geq 4$,
apply (B3) on $(G,k)$ and $v$.
\item[(S6)]
If $v$ is a vertex with $\Deg{v} \geq 4$,
apply (B1) on $(G,k)$ and $v$.
\item[(S7)]
If $v$ is a vertex with $\Deg{v} = 2$ and a neighbor $u$ of $v$ is inside
a triangle, apply (B3) on $(G,k)$ and $w$, where $w$ is the neighbor of $v$
other than $u$.
Apply (B4) on $(G-w,k-1)$ and $u$.
\item[(S8)]
If $v$ is a vertex with $\Deg{v} = 2$ and a neighbor $w$ of $v$ has a neighbor
$w_1$ with degree~3, apply (B3) on $(G,k)$ and $w$.
Apply (B2) on $(G-w,k-1)$ and $u$, where $u$ is the neighbor of $v$ other than
$w$.
\item[(S9)] Let $S$ be a connected component of $G$ such that $G[S]$ is a
bipartite graph with parts $S_1,S_2$ such that $\Deg{v} = 2$ for every
$v\in S_1$ and $\Deg{v} = 3$ for every $v\in S_2$.
Apply the rule $(G,k)\to \Rule{G-S}{k-|S_2|}{S_2}$.
\item[(S10)]
Pick an arbitrary vertex $v$. Apply (B1) on $(G,k)$ and $v$.
\end{enumerate}
We note that the order of Step~(S4) is different in the algorithm of Xiao and Kou, but this does not affect the correctness of the algorithm.

The maximum branching factor of Steps (S1)--(S9) is approximately 1.749.
Additionally, Step~(S10) does not affect the asymptotic time complexity of
the algorithm since it is applied at most one in each path in the recursion
tree.
Therefore, the time complexity of the algorithm is $O^*(1.749^k)$.

\section{New algorithm}\label{sec:new}

In this section we describe our algorithm.
Our algorithm is the same as the algorithm of Xiao and Kou except that we
modify Steps~(S5) and~(S8).
Each of these steps is replaced by several sub-steps.
The maximum branching factor of these sub-steps is approximately 1.713.
The branching factors of the other steps in the algorithm of Xiao and Kou
are at most 1.710.
Therefore, the time complexity of the algorithm is $O^*(1.713^k)$.

\subsection{Step (S5)}\label{sec:S5}
Since steps~(S1)--(S4) cannot be applied, the graph $G$ has the following
properties.
\begin{inparaenum}[(P1)]
\item\label{prop:no-degree-1}
Every vertex has degree at least 2.
\item\label{prop:degree-2-neighbors}
For every vertex with degree 2, its neighbors have degree at least 3
and they are not adjacent.
\item\label{prop:no-dominating}
There are no dominating vertices.
\end{inparaenum}

In the following sub-steps, a common condition for applying the steps is
that the graph $G$ has a weakly dominating vertex $v$ with $\Deg{v} \geq 4$.
Each sub-step (except the last) has an additional condition which needs to
be satisfied.
See Figure~\ref{fig:S5} for graphs satisfying the different conditions
of the sub-steps.

\begin{figure}
\subfigure[(S5-1)]{\includegraphics{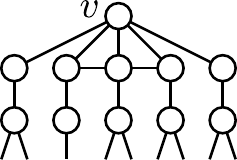}}
\hfill
\subfigure[(S5-2)]{\includegraphics{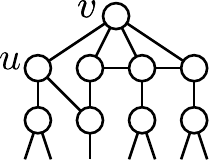}}
\hfill
\subfigure[(S5-3)]{\includegraphics{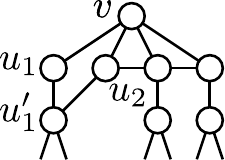}}
\hfill
\subfigure[(S5-4)]{\includegraphics{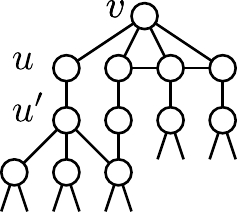}}
\hfill
\subfigure[(S5-5)]{\includegraphics{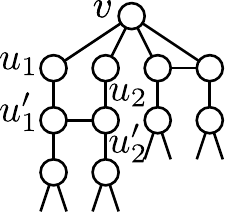}}
\hfill
\subfigure[(S5-6)]{\includegraphics{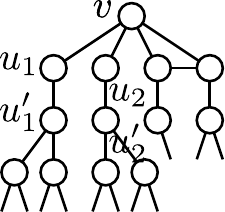}}
\hfill
\subfigure[(S5-7)]{\includegraphics{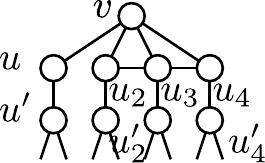}}
\hfill
\subfigure[(S5-8)]{\includegraphics{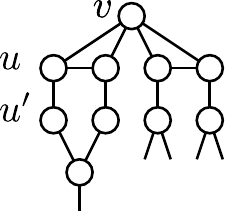}}
\hfill
\subfigure[(S5-9)]{\includegraphics{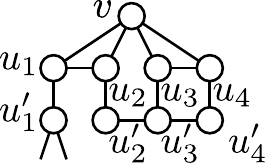}}
\hfill
\subfigure[(S5-11)]{\includegraphics{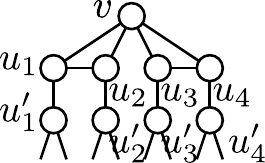}}
\caption{Example for the sub-steps of (S5).}
\label{fig:S5}
\end{figure}

\paragraph{Step (S5-1)}
If $\Deg{v} \geq 5$, apply (B3) on $(G,k)$ and $v$.
Recall that (B3) generates the instances
$(G-v,k-1)$ and
$\{(G-N[\{u,v\}],k-|N(\{u,v\})|)\}_{u\in N(v)}$.
By Property~(P\ref{prop:no-dominating}) and Claim~\ref{clm:dominating},
$|N(\{u,v\})| \geq \Deg{v}$.
Therefore, the recurrence of this step is
\[ T(k) \leq T(k-1)+\Deg{v}\cdot T(k-\Deg{v}). \]
The worst case of the recurrence is when $\Deg{v} = 5$, and the branching
factor is at most 1.660.

In the remaining sub-steps we have $\Deg{v} = 4$ since (S5-1) cannot be applied.
\paragraph{Step (S5-2)}
If there is a neighbor $u'$ of $v$ that has at least two neighbors that are not
neighbors of $v$, apply (B3) on $(G,k)$ and $v$.
By Claim~\ref{clm:dominating},
$|N(\{u,v\})| \geq \Deg{v} = 4$ for every $u\in N(v)$.
Additionally, $|N(\{u',v\})| \geq 5$.
Therefore, the recurrence of this step is (recall that $\Deg{v} = 4$)
\[ T(k) \leq T(k-1)+3T(k-4)+T(k-5) \]
and the branching factor is approximately 1.713.

In the remaining sub-steps we have that every vertex $u\in N(v)$ has
exactly one neighbor that is not in $N(v)$. This neighbor will be denoted
by $u'$.
\paragraph{Step (S5-3)}
If there are vertices $u_1,u_2\in N(v)$ such that $u'_1 = u'_2$,
apply (B3) on $(G,k)$ and $v$.
In this case, $N[\{u_1,v\}] = N[\{u_2,v\}]$ and
$|N(\{u_1,v\})| = |N(\{u_2,v\})|$.
Therefore, the instance
$(G-N[\{u_1,v\}],k-|N(\{u_1,v\})|)$ is the same instance as
$(G-N[\{u_2,v\}],k-|N(\{u_2,v\})|)$.
The algorithm recurses on only one copy of this instance. Thus,
the recurrence of this step is
\[ T(k)\leq T(k-1)+3T(k-4) \]
and the branching factor is approximately 1.659.

\paragraph{Step (S5-4)}
If there is a vertex $u\in N(v)$ that does not have neighbors in $N(v)$ and
$\Deg{u'} \geq 4$,
apply (B3) on $(G,k)$ and $v$.
Note that $\Degb{G-v}{u} = 1$. Therefore, $u'$ dominates $u$ in $G-v$.
Additionally, $\Degb{G-v}{u'} = \Deg{u'}$.
Apply (B2) on $(G-v,k-1)$ and $u'$, which generates the instances
$(G-\{v,u'\},k-2)$ and $(G-(\{v\}\cup N_{G-v}[u']),k-\Degb{G-v}{u'})$.
The recurrence of this step is (the worst case is $\Deg{u'}=4$)
\[ T(k)\leq (T(k-2)+T(k-4))+4T(k-4) \]
and the branching factor is approximately 1.671.

\paragraph{Step (S5-5)}
If there are vertices $u_1,u_2 \in N(v)$ such that each of these vertices
does not have neighbors in $N(v)$ and $(u'_1,u'_2)\in E$,
apply (B3) on $(G,k)$ and $v$.
Since (S5-4) cannot be applied, $\Deg{u'_1},\Deg{u'_2} \leq 3$.
By Property~(P\ref{prop:degree-2-neighbors}), $\Deg{u'_1} = \Deg{u'_2} = 3$.
As in (S5-4), $u'_1$ dominates $u_1$ in $G-v$.
Apply (B2) on $(G-v,k-1)$ and $u'_1$, which generates the instances
$(G'=G-\{v,u'_1\},k-2)$ and $(G-(\{v\}\cup N_{G-v}[u'_1]),k-3)$.
Note that $\Degb{G'}{u_2} = 1$ and $\Degb{G'}{u'_2} = \Deg{u'_2}-1 = 2$
(since $u'_1$ is a neighbor of $u'_2$ in $G$ but not in $G'$, and the other
two neighbors of $u'_2$ in $G$ are also neighbors of $u'_2$ in $G'$).
Next, apply (S2) on $(G',k-2)$ and $u_2$, which generates the instance
$(G' - N_{G'}[u'_2],k-3)$.
The recurrence of this step is
\[ T(k)\leq 2T(k-3)+4T(k-4) \]
and the branching factor is approximately 1.643.

\paragraph{Step (S5-6)}
If there are vertices $u_1,u_2 \in N(v)$ such that each of these vertices does
not have neighbors in $N(v)$,
apply (B3) on $(G,k)$ and $v$.
As in (S5-5), $\Deg{u'_1} = \Deg{u'_2} = 3$ and $u'_1$ dominates $u_1$ in $G-v$.
Apply (B2) on $(G-v,k-1)$ and $u'_1$, which generates the instances
$(G'=G-\{v,u'_1\},k-2)$ and $(G-(\{v\}\cup N_{G-v}[u'_1]),k-3)$.
We have that $u'_2$ dominates $u_2$ in $G'$ and
$\Degb{G'}{u'_2} = \Deg{u'_2} = 3$.
Apply (B2) on $(G',k-2)$ and $u'_2$, which generates the instances
$(G' - \{u'_2\},k-3)$ and
$(G' - N_{G'}[u'_2],k-4)$.
The recurrence of this step is
\[ T(k)\leq (2T(k-3)+T(k-4))+4T(k-4) \]
and the branching factor is approximately 1.703.

\paragraph{Step (S5-7)}
If there is exactly one vertex $u \in N(v)$ that does not have a neighbor in
$N(v)$, apply (B3) on $(G,k)$ and $v$.
Next, apply (B2) on $(G-v,k-1)$ and $u'$, which generates the instances
$(G' = G-\{v,u'\},k-2)$ and $(G-(\{v\}\cup N[u']),k-\Deg{u'})$.
Let $u_2,u_3,u_4$ be the remaining neighbors of $v$, numbered such that
$(u_2,u_3),(u_3,u_4) \in E$.
Note that $\{v,u'\}\cap\{u_1,u_2,u_3,u'_1,u'_2,u'_3\} = \emptyset$.
Apply (B3) on $(G',k-2)$ and $u_3$ ($u_3$ weakly dominates $u_2$ in $G'$),
which generates the instances
$(G'-\{u_3\},k-3)$ and
$(G'-N_{G'}[\{w,u_3\}]),k-2-|N_{G'}(\{w,u_3\})|)$ for every
$w\in N_{G'}(u_3) = \{u_2,u'_3,u_4\}$.
Note that for $w \in \{u_2,u_4\}$,
$|N(\{w,u_3\})| = 3$ (since $N_{G'}(\{u_2,u_3\})=\{u'_2,u'_3,u_4\}$
and $N_{G'}(\{u_4,u_3\})=\{u_2,u'_3,u_4'\}$).
Consider the instance
$(G'' = G'-N_{G'}[\{u'_3,u_3\}],k'' = k-2-|N_{G'}(\{u'_3,u_3\})|)$
which is generated by the application of (B3) on $(G',k-2)$.
We claim that if $(G'',k'')$ is a yes instance then
$(G-N[\{u,v\}],k-4)$ (which is generated by the application of (B3) on
$(G,k)$) is also a yes instance.
To prove the claim, note that if $C$ is 3-path vertex cover of $G''$ of size
at most $k'' =  k-2-|N_{G'}(\{u'_3,u_3\})| = k-3-\Deg{u'_3}$ then 
$C\cup (N(u'_3)\setminus\{u_3\})$ is
a 3-path vertex cover of $G-N[\{u,v\}]$ of size at most
$|C|+(\Deg{u'_3}-1)\leq k-4$.
It follows that the algorithm does not need to recurse on $(G'',k'')$.
Therefore, the recurrence of this step is
\[ T(k)\leq (2T(k-3)+2T(k-5))+4T(k-4) \]
and the branching factor is approximately 1.713.


We now claim that in the remaining sub-steps we have that every $u\in N(v)$ has
exactly one neighbor in $N(v)$.
Clearly, since (S5-6) and~(S5-7) cannot be applied, every $u\in N(v)$ has
at least one neighbor in $N(v)$.
Suppose conversely that there is $u_2 \in N(v)$ which is adjacent to
two distinct vertices $u_1,u_3 \in N(v)$.
Let $u_4$ be the remaining vertex in $N(v)$.
$u_4$ is not a neighbor of $u_2$ since otherwise $u_2$ dominates $v$,
contradicting Property~(P\ref{prop:no-dominating}).
Since $u_4$ has at least one neighbor in $N(v)$, we have that $u_4$ is
adjacent to either $u_1$ or $u_3$.
Without loss of generality assume that $u_4$ is adjacent to $u_3$.
We have that $u_2$ weakly dominates $v$. Additionally, $\Deg{u_2} = 4$
and $u_3 \in N(u_2)$ has two neighbors that are not in $N(u_2)$:
$u'_3$ and $u_4$.
This contradicts the fact that (S5-2) cannot be applied.

\paragraph{Step (S5-8)}
If there is a vertex $u \in N(v)$ whose connected component $S$ in $G-v$ is
a cordless cycle, apply (B3) on $(G,k)$ and $v$.
Then apply the rule $(G-v,k-1) \to (G-(\{v\}\cup S),k-1-\gamma)$, where $\gamma=\lceil|S|/3\rceil$ is the
minimum size of a 3-path vertex of $G[S]$.
We have that $|S|\geq 4$ (since $\{u,u',u_2,u'_2\}\subseteq S$, where
$u_2$ is the unique neighbor of $u$ in $N(v)$), so
$\gamma \geq 2$.
The recurrence of this step is (the worst case is $\gamma = 2$)
\[ T(k)\leq T(k-3)+4T(k-4) \]
and the branching factor is approximately 1.534.

For the remaining sub-steps, denote the neighbors of $v$ by
$u_1,u_2,u_3,u_4$, where $(u_1,u_2)\in E$ and $(u_3,u_4)\in E$ are the
only edges between these vertices.
\begin{lemma}\label{lem:chain-cases}
There is a chain $x,x_1,x_2,x_3$ in $G-v$ such that the graph
$G-\{v,x,x_1,x_2,x_3\}$ contains one of the following structures:
(1) A connected component $S$ of size at least 4 which is a cordless path
or a cordless cycle.
(2) A chain $y,y_1,y_2,y_3$.
\end{lemma}
\begin{proof}
We have that $\Degb{G-v}{u_i} = 2$ for all $i$, and
$\Degb{G-v}{w} = \Deg{w}$ for all $w \notin N[v]$.
If $\Deg{u'_2} \geq 3$, $x,x_1,x_2,x_3 = u'_2,u_2,u_1,u'_1$ is a chain
in $G-v$.
Otherwise, if $\Deg{u'_1} \geq 3$, $x,x_1,x_2,x_3 = u'_1,u_1,u_2,u'_2$ is
a chain in $G-v$.
Now suppose that $\Deg{u'_1},\Deg{u'_2} \leq 2$.
By Property~(P\ref{prop:no-degree-1}), $\Deg{u'_1} = \Deg{u'_2} = 2$.
Let $z$ be neighbor of $u'_2$ other than $u_2$.
By Property~(P\ref{prop:degree-2-neighbors}), $\Deg{z} \geq 3$.
Therefore, $z,u'_2,u_2,u_1$ is a chain in $G-v$.
If $z\notin\{u'_3,u'_4\}$ define $x,x_1,x_2,x_3 = z,u'_2,u_2,u_1$.
Otherwise, suppose without loss of generality that $z = u'_3$, and
$x,x_1,x_2,x_3 = u'_3,u_3,u_4,u'_4$ is a chain in $G-v$.

Let $G' = G-\{v,x,x_1,x_2,x_3\}$.
In the first three cases above,
$\{v,x,x_1,x_2,x_3\}\cap \{u_3,u_4,u'_3,u'_4\} = \emptyset$.
Since $u_3,u_4$ are adjacent vertices with no common neighbor and
$\Degb{G'}{u_3} = \Degb{G'}{u_4} = 2$, the lemma follows from
Claim~\ref{clm:chain}.
In the fourth case above,
$\{v,x,x_1,x_2,x_3\}\cap \{u_1,u_2,u'_1,u'_2\} = \emptyset$ and 
the lemma follows again from Claim~\ref{clm:chain}.
\end{proof}
\begin{lemma}
There are no dominating vertices in $G-v$.\label{lem:G-v}
\end{lemma}
\begin{proof}
By Property~(P\ref{prop:no-dominating}), if $G-v$ has a dominating vertex
then either the dominating vertex is $u_i$ for some $i$,
or the dominated vertex is $u_i$ for some $i$.
The neighbors of $u_i$ are $u'_i$ and $u_j$ for some $j$.
Since $u'_i$ and $u_j$ are not adjacent, $u_i$ is not dominated by either of
these vertices.
Moreover, both $u'_i$ and $u_j$ have neighbors that are not adjacent to $u_i$,
so $u_i$ does not dominate these vertices.
\end{proof}

\paragraph{Step (S5-9)}
If case~(1) of Lemma~\ref{lem:chain-cases} occurs,
apply (B3) on $(G,k)$ and $v$.
Apply (S4) on $(G-v,k-1)$ and the chain $x,x_1,x_2,x_3$, which generates the
instances $(G' = G-\{v,x,x_1,x_2,x_3\},k-3)$ and
$(G-(\{v\}\cup N_{G-v}[\{x',x\}]),k-1-|N_{G-v}(\{x',x\})|)$
for every $x'\in N_{G-v}(x)$.
By Claim~\ref{clm:dominating} and Lemma~\ref{lem:G-v}, for every $x' \in N(x)$,
$|N_{G-v}(\{x',x\})| \geq \Degb{G-v}{x} = \Deg{x} \geq 3$.
Next, apply the rule
$(G',k-3) \to (G'-S,k-3-\gamma)$, where $\gamma \geq 1$ is the minimum size
of a 3-path vertex cover of $S$.
The recurrence of this step is (the worst case is $\gamma = 1$)
\[ T(k)\leq 4T(k-4)+4T(k-4) \]
and the branching factor is approximately 1.682.

In the following two sub-steps we have that case~(3) of
Lemma~\ref{lem:chain-cases} occurs.
\paragraph{Step (S5-10)}
If $y$ is a dominating vertex,
apply (B3) on $(G,k)$ and $v$.
Apply (S4) on $(G-v,k-1)$ and the chain $x,x_1,x_2,x_3$, which generates the
instances $(G' = G-\{v,x,x_1,x_2,x_3\},k-3)$ and
$(G-(\{v\}\cup N_{G-v}[\{x',x\}]),k-1-|N_{G-v}(\{x',x\})|)$
for every $x'\in N_{G-v}(x)$.
Apply (B2) on $(G',k-3)$ and $y$ which generates the instances $(G'-y,k-4)$ and
$(G'-N_{G'}[y],k-2-\Degb{G'}{y})$.
The recurrence of this step is
\[ T(k)\leq (4T(k-4)+T(k-5))+4T(k-4) \]
and the branching factor is approximately 1.712.

\paragraph{Step (S5-11)}
In this step $y$ is not a dominating vertex.
Apply (B3) on $(G,k)$ and $v$.
Apply (S4) on $(G-v,k-1)$ and the chain $x,x_1,x_2,x_3$, which generates the
instances $(G' = G-\{v,x,x_1,x_2,x_3\},k-3)$ and
$(G-(\{v\}\cup N_{G-v}[\{x',x\}]),k-1-|N_{G-v}(\{x',x\})|)$
for every $x'\in N_{G-v}(x)$.
Then, apply (S4) on $(G',k-3)$ and $y,y_1,y_2,y_3$,
which generates the instances
$(G'-\{y,y_1,y_2,y_3\},k-5)$ and
$(G'-N_{G'}[\{y',y\}]),k-3-|N_{G'}(\{y',y\})|)$
for every $y'\in N_{G'}(y)$.
By Claim~\ref{clm:dominating} and since Step~(S5-10) cannot be applied,
$|N_{G'}(\{y',y\})| \geq \Degb{G'}{y}$ for every $y'\in N_{G'}(y)$.
The recurrence of this step is
\[ T(k)\leq (3T(k-4)+T(k-5)+3T(k-6))+4T(k-4) \]
and the branching factor is approximately 1.713.

\subsection{Step (S8)}\label{sec:S8}
Since steps~(S1)--(S7) cannot be applied, the graph $G$ has the following
properties, in addition to the properties stated in Section~\ref{sec:S5}.
\begin{inparaenum}[(P1)]
\setcounter{enumi}{3}
\item\label{prop:degree-2-3}
Every vertex has degree 2 or 3.
\item\label{prop:degree-2-triangle}
A vertex with degree 2 is not adjacent to a vertex in a triangle.
\end{inparaenum}

In the following sub-steps, a common condition for applying the steps is
that there is a vertex $v$ such that $\Deg{v} = 2$ and there is a neighbor
$w$ of $v$ and a neighbor $w_1$ of $w$ with $\Deg{w_1} = 3$.
Let $u$ be the other neighbor of $v$.
Due to Property~(P\ref{prop:degree-2-neighbors}),
$\Deg{w} = \Deg{u} = 3$.
Denote the neighbors of $w$ and $u$ which are not $v$ by $w_1,w_2$ and 
$u_1,u_2$, respectively.
See Figure~\ref{fig:S8} for graphs satisfying the different conditions
of the sub-steps.

\begin{figure}
\subfigure[(S8-1)]{\includegraphics{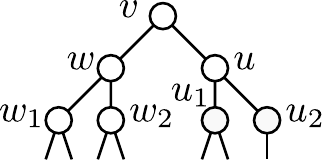}}
\hfill
\subfigure[(S8-2)]{\includegraphics{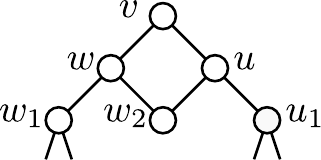}}
\hfill
\subfigure[(S8-3)]{\includegraphics{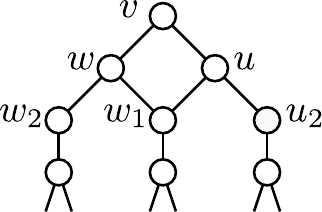}}
\hfill
\subfigure[(S8-4)]{\includegraphics{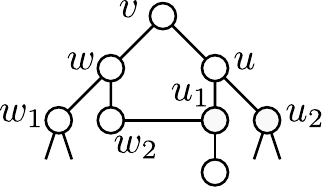}}
\hfill
\subfigure[(S8-5)]{\includegraphics{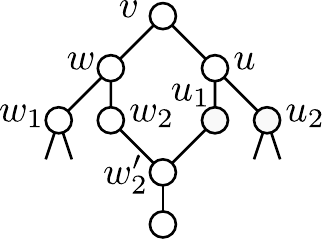}}
\hfill
\subfigure[(S8-6)]{\includegraphics{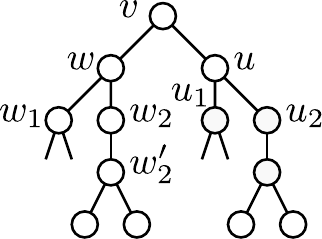}}

\caption{Example for the sub-steps of (S8).}
\label{fig:S8}
\end{figure}

\paragraph{Step (S8-1)}
If $\Deg{w_2} = 3$,
apply (B3) on $(G,k)$ and $w$ which generates the instances
$(G-w,k-1)$ and $(G-N[\{w',w\}],k-|N(\{w',w\})|)$ for every
$w' \in N(w) = \{w_1,w_2,v\}$.
Apply (B2) on $(G-w,k-1)$ and $u$ which generates the instances
$(G-\{u,w\},k-2)$ and $(G-(\{w\}\cup N_{G-w}[u]),k-3)$.
We have that $|N(\{w_1,w\})| = 4$ since $w,w_1$ are adjacent vertices
with degree~3 that do not have common neighbor
(due to Property~(P\ref{prop:degree-2-triangle})).
By the same arguments $|N(\{w_2,w\})| = 4$.
Additionally, $v$ and $w$ do not have common neighbor
(due to Property~(P\ref{prop:degree-2-neighbors})),
thus $|N(\{v,w\})| = 3$.
Therefore, the recurrence of this step is
\[ T(k)\leq T(k-2)+2T(k-3)+2T(k-4) \]
and the branching factor is approximately 1.696.

Due to Property~(P\ref{prop:degree-2-3}),
in the following sub-steps we have $\Deg{w_2} = 2$.
Additionally, at least one of $u_1$ and $u_2$ has degree $2$
(otherwise (S8-1) can be applied on $v$) and assume that $\Deg{u_2} = 2$
(note that the degree of $u_1$ is either~3 or~2).
Let $w'_2$ be the neighbor of $w_2$ which is not $w$, and
let $u'_2$ be the neighbor of $u_2$ which is not $u$.
We have that $\Deg{w'_2} = \Deg{u'_2} = 3$
(due to Property~(P\ref{prop:degree-2-neighbors})).
Let $\delta_1$ be the minimum length of a path in $G-\{u,w\}$ connecting
a vertex from $\{w_1,w_2\}$ to a vertex from $\{u_1, u_2\}$ such that
one endpoint of the path has degree~3 (in $G$) and the other endpoint has
degree~2 (in $G$).
Let $\delta_2$ be the minimum length of a path in $G-\{u,w\}$ connecting
$\{w_2\}$ to a vertex from $\{u_1, u_2\}$ such that both endpoints of the path
have degree~2 (in $G$).
Let $\delta = \min(\delta_1,\delta_2)$.

\paragraph{Step (S8-2)}
Suppose that $\delta = 0$.
Note that by definition $\delta_1 \neq 0$, so we have $\delta_2 = 0$.
In other words, without loss of generality, $w = u_2$.
The set $\{v,u,w,w_2\}$ induces a cordless cycle. The vertices of this cycle
that are adjacent to vertices outside the cycle are $w$ and $u$.
Therefore, there is a minimum 3-path vertex cover of $G$ which contains
$w,v$ and does not contain $v,w_2$.
Therefore, in this case apply the rule $(G,k)\to (G-\{v,u,w,w_2\},k-2)$.

\paragraph{Step (S8-3)}
If $w_1 = u_1$, as in Step~(S8-1),
apply (B3) on $(G,k)$ and $w$ and then apply (B2) on $(G-w,k-1)$ and $u$.
We claim that if the instance
$(G-N[\{w_1,w\}],k-|N(\{w_1,w\})|)=(G-N[\{w_1,w\}],k-4)$ (which is generated
by the application of (B3)) is a yes instance then
the instance $(G-\{u,w\},k-2)$ (which is generated by the application of (B2))
is a yes instance.
To prove the correctness of this claim, note that if $C$ is a 3-path vertex
cover of $G-N[\{w_1,w\}]$ of size at most $k-4$, then
$C\cup \{w_2,x\}$, where $x$ is the neighbor of $w_1$ other than $w$ and $u$,
is a 3-path vertex cover of $G-\{u,w\}$ of size at most $|C|+2 \leq k-2$.
Therefore, the algorithm does not recurse on $(G-N[\{w_1,w\}],k-4)$.
The recurrence of this step is
\[ T(k)\leq T(k-2)+3T(k-3) \]
and the branching factor is approximately 1.672.

\paragraph{Step (S8-4)}
If $\delta = 1$, perform the following.
Note that $\delta_2 \neq 1$ (by the definition of $\delta_2$ and the fact
that $\Deg{w'_2} = 3$), so  $\delta_1 = 1$.
Without loss of generality, $\delta_1$ is a length of a path in $G-\{u,w\}$
between $w_2$ and $u_1$.
As in Step~(S8-1),
apply (B3) on $(G,k)$ and $w$ and then apply (B2) on $(G-w,k-1)$ and $u$.
Let $x$ be the neighbor of $u_1$ other than $u$ and $w_2$.
Note that $x \neq w$ (otherwise the two neighbors $w$ and $u_1$ of $w_2$
are adjacent, contradicting Property~(P\ref{prop:degree-2-neighbors})).
Therefore, $\Degb{G-\{u,w\}}{u_1} = 2$.
Additionally, $\Degb{G-\{u,w\}}{w_2} = 1$.
Therefore, apply Step~(S2) on $(G'=G-\{u,w\},k-2)$
(which is generated by the application of (B2))
and $w_2$, which generates the instance $(G'-N_{G'}[u_1],k-3)$.
The recurrence of this step is
\[ T(k)\leq 4T(k-3)+T(k-4) \]
and the branching factor is approximately 1.664.

\paragraph{Step (S8-5)}
If $\delta = 2$, perform the following.
Assume without loss of generality that the value of $\delta$
is obtained due to a path between $w_2$ and either $u_1$ or $u_2$.
As in Step~(S8-1),
apply (B3) on $(G,k)$ and $w$ and then apply (B2) on $(G-w,k-1)$ and $u$.
Next, apply (B2) on $(G'=G-\{u,w\},k-2)$ and $w'_2$, which generates the
instances $(G'-w'_2,k-3)$ and $(G'-N[w'_2],k-4)$.
We claim that if $(G'-N[w'_2],k-4)$
(which is obtained by the second application of (B2))
is a yes instance then $(G-(\{w\}\cup N_{G-w}[u]),k-3)$
(which is generated by the first application of (B2))
is a yes instance.
To prove the correctness of this claim, note that if $C$ is a 3-path vertex
cover of $G'-N[w'_2]$ of size at most $k-4$, then
$C\cup \{x\}$, where $x$ is the neighbor of $w'_2$ other than $w_2$ and $u_2$,
is a 3-path vertex cover of $G-(\{w\}\cup N_{G-w}[u])$ of size at most
$|C|+1 \leq k-3$.
Therefore, the algorithm does not recurse on the former instance.
The recurrence of this step is
\[ T(k)\leq 4T(k-3)+T(k-4) \]
and the branching factor is approximately 1.664.

\paragraph{Step (S8-6)}
As in Step~(S8-1),
apply (B3) on $(G,k)$ and $w$ and then apply (B2) on $(G-w,k-1)$ and $u$.

Consider the instance $(G'=G-N[\{w_1,w\}],k-4)$ that is generated by the
application of (B3).
Since $w,v$ are not adjacent (due to Property~(P\ref{prop:degree-2-neighbors})),
$\delta \geq 3$, and $w_1 \neq u_1$, we have that
$\{u,u_2,u'_2\}\cap N[\{w_1,w\}] = \emptyset$.
If $\Degb{G'}{u} = 1$, then since $\Degb{G'}{u_2} = 2$,
apply (S2) on $(G',k-4)$. 
Otherwise, $\Degb{G'}{u} = 2$ and by Claim~\ref{clm:chain}
either the connected component $S$ of $u$ in $G'$ is a path or a
cycle (and no additional edges), or $S$ contains a chain $x,x_1,x_2,x_3$.
In the former case the algorithm computes the minimum size $\gamma$ of a
3-path vertex cover of $S$, and apply the rule
$(G',k-4) \to (G'-S,k-4-\gamma)$ (note that $|S|\geq 3$ so $\gamma \geq 1$).
In the latter case, if $x$ is a dominating vertex, apply (B2) on $(G',k-4)$
and $x$. 
Otherwise, apply (S4) on $(G',k-4)$ and $x,x_1,x_2,x_3$, which generates
the instances $(G'-\{x,x_1,x_2,x_3\},k-6)$ and
$(G'-N[\{x',x\}],k-4-|N_{G'}(\{x',x\})|)$ for every $x'\in N_{G'}(x)$.

We also handle the instance $(G-N[\{w_2,w\}],k-3)$ in a similar way.

We now consider the instances $(G-\{u,w\},k-2)$ and
$(G''=G-(\{w\}\cup N_{G-w}[u]),k-3)$ that were generated by the application of
(B2) on $(G-w,k-1)$.
Apply (B2) on $(G-\{u,w\},k-2)$ and $w'_2$, which generates the instances
$(G-\{u,w,w'_2\},k-3)$ and $(G-(\{u,w\}\cup N[w'_2]),k-4)$.
For the instance $(G'',k-3)$ we have that
$N[w'_2]\cap(\{w\}\cup N_{G-w}[u])=\emptyset$
(since $w,v$ are not adjacent, $w_1,w_2$ are not adjacent, and $\delta \geq 3$).
Therefore, apply (B2) on $(G'',k-3)$ and $w'_2$ which generates the instances
$(G''-w'_2,k-4)$ and $(G''-N_{G''}[w'_2],k-5)$.
The recurrence of this step is (the worst case occurs when (S4) is applied on
$(G',k-4)$ and on $(G-N[\{w_2,w\}],k-3)$).
\[ T(k)\leq (T(k-3)+2T(k-4)+T(k-5))+
(T(k-3)+T(k-5)+4T(k-6)+3T(k-7))\]
and the branching factor is approximately 1.711.

\bibliographystyle{plain}
\bibliography{parameterized}

\end{document}